\documentclass[11pt]{article}
\usepackage[utf8]{inputenc}

\iffalse
\usepackage[accepted]{icml2021_arxiv}
\else
\usepackage{fullpage}
\usepackage{algorithm}
\usepackage{algorithmic}
\fi

\usepackage{xcolor}
\usepackage{amsmath}
\usepackage{amssymb}
\usepackage{amsthm}
\usepackage[]{algorithm}
\usepackage{cleveref}
\usepackage{graphicx,nicefrac}
\usepackage{subcaption}
\usepackage[font={small}]{caption}
\usepackage{thm-restate}
\usepackage[normalem]{ulem}
\usepackage{float}
\usepackage{balance}

\usepackage{enumitem}
\usepackage{xcolor}
\newcommand{\nc}{\newcommand}
\newcommand{\DMO}{\DeclareMathOperator}
\DeclareMathAlphabet\mathbfcal{OMS}{cmsy}{b}{n}

\newcommand{\silly}{\hspace*{1em}}

\newcommand{\Procedure}[2]{\STATE {{\bf procedure} {\sc #1}}} 
\newcommand{\State}{\STATE \silly}

\newcommand{\Return}{{\bf return \/}}

\newcommand{\EndProcedure}{}
\newcommand{\EndIf}{}

\allowdisplaybreaks

\nc{\MS}{\mathcal{S}}
\nc{\MP}{\mathcal{P}}
\nc{\MR}{\mathcal{R}}
\nc{\cM}{\mathcal{M}}
\nc{\cS}{\mathcal{S}}
\nc{\cI}{\mathcal{I}}
\nc{\cA}{\mathcal{A}}
\nc{\tcA}{\tilde{\cA}}

\nc{\MZ}{\mathcal{Z}}

\DMO{\Binom}{Binom}
\newcommand{\E}{\mathbb{E}}
\DMO{\Var}{Var}

\newcommand{\bh}{\mathbf{h}}
\newcommand{\tbh}{\tilde{\bh}}

\newcommand{\bX}{\mathbf{X}}
\nc{\tbx}{\tilde{\bx}}
\nc{\tbX}{\tilde{\bX}}
\nc{\tZ}{\tilde{Z}}
\nc{\tz}{\tilde{z}}
\newcommand{\bU}{\mathbf{U}}
\nc{\tbU}{\tilde{\bU}}
\newcommand{\bT}{\mathbf{T}}
\nc{\tbT}{\tilde{\bT}}
\newcommand{\bD}{\mathbf{D}}
\nc{\tbD}{\tilde{\bD}}

\newcommand{\bx}{\mathbf{x}}

\newcommand{\R}{\mathbb{R}}

\nc{\BN}{\mathbb{N}}
\newcommand{\Z}{\mathbb{Z}}
\nc{\BZ}{\mathbb{Z}}

\newcommand{\bone}{\mathbf{1}}

\newcommand{\bA}{\mathbf{A}}
\nc{\tbA}{\tilde{\bA}}
\renewcommand{\th}{\tilde{h}}

\newcommand{\cD}{\mathcal{D}}

\newcommand{\cDcentral}{\cD^{\mathrm{noise}}}
\newcommand{\cDflood}{\cD^{\mathrm{flood}}}
\newcommand{\inp}{\mathrm{input}}
\newcommand{\cF}{\mathcal{F}}

\DeclareMathOperator{\supp}{supp}

\DeclareMathOperator{\DLap}{DLap}

\DeclareMathOperator{\Poi}{Poi}
\DeclareMathOperator{\NB}{NB}
\DeclareMathOperator{\Geo}{Geo}

\newcommand{\eps}{\epsilon}
\newcommand{\tO}{\widetilde{O}}
\newcommand{\cDinp}[1]{\cD^{\inp, #1}}

\newcommand{\DP}{\mathrm{DP}}
\newcommand{\typeOfDP}[1]{\DP_{\mathrm{ #1}}}

\newcommand{\shuffledDP}{\typeOfDP{shuffle}}

\newcommand{\randomizer}{\textsc{CorrNoiseRandomizer}}
\newcommand{\analyzer}{\textsc{CorrNoiseAnalyzer}}

\newtheorem{theorem}{Theorem}

\newtheorem{lemma}[theorem]{Lemma}

\newtheorem{definition}[theorem]{Definition}
\newtheorem{condition}[theorem]{Condition}

\newtheorem{corollary}[theorem]{Corollary}

\title{Pure-DP Aggregation in the Shuffle Model:\\
Error-Optimal and Communication-Efficient}

\author{
Badih Ghazi\thanks{Google Research, Mountain View. Email: \texttt{badihghazi@gmail.com}.}
\and
Ravi Kumar\thanks{Google Research, Mountain View. Email: \texttt{ravi.k53@gmail.com}.}
\and
Pasin Manurangsi\thanks{Google Research, Thailand. Email: \texttt{pasin@google.com}.}
}
\date{\today}

\begin{document}

\maketitle

\begin{abstract}


We obtain a new protocol for binary counting in the $\eps$-$\shuffledDP$ model with error $O(1/\eps)$ and expected communication  $\tO\left(\frac{\log n}{\eps}\right)$ messages per user.
Previous protocols incur either an error of $O(1/\epsilon^{1.5})$ with  $O_\eps(\log{n})$ messages per user (Ghazi et al., ITC 2020) or an error of $O(1/\epsilon)$ with $O_\eps(n^{2.5})$  messages per user (Cheu and Yan, TPDP 2022).  Using the new protocol, we obtained improved $\eps$-$\shuffledDP$ protocols for real summation and histograms.
\end{abstract}

\section{Introduction}



Differential privacy (DP)~\cite{dwork2006calibrating} is a widely accepted notion used for bounding and quantifying an algorithm’s  leakage of personal information.  Its most basic form, known as \emph{pure}-DP, is governed by a single parameter $\epsilon > 0$, which bounds the leakage of the algorithm. Specifically, a randomized algorithm $A(\cdot)$ is said to be \emph{$\epsilon$-DP} if for any subset $S$ of output values, and for any two datasets $D$ and $D'$ differing on a single user’s data, it holds that $\Pr[A(D) \in S] \le e^{\epsilon} \cdot \Pr[A(D') \in S]$.
In settings where pure-DP is not (known to be) possible, a common relaxation is the so-called \emph{approximate} $(\epsilon, \delta)$-DP \cite{dwork2006our}, which has an additional parameter $\delta \in [0,1]$. In this case, the condition becomes: $\Pr[A(D) \in S] \le e^{\epsilon} \cdot \Pr[A(D') \in S] + \delta$.

Depending on the trust assumptions, three models of DP are commonly studied. The first is the \emph{central} model, where a trusted curator is assumed to hold the raw data and required to release a private output. (This goes back to the first work \cite{dwork2006calibrating} on DP.)
The second is the \emph{local} model \cite{evfimievski2003limiting, dwork2006calibrating, kasiviswanathan2008what}, where each user’s message is required to be private.
The third is the \emph{shuffle} model \cite{bittau17, CheuSUZZ19,erlingsson2019amplification}, where the users’ messages are routed through a trusted shuffler, which is assumed to be non-colluding with the curator, and which is expected to randomly permute the messages incoming from the different users ($\shuffledDP$).  More formally, a protocol $P = (R, S, A)$ in the shuffle model consists of $3$ procedures: (i) a local randomizer $R(\cdot)$ that takes as input the data of a single user and outputs one or more messages, (ii) a shuffler $S(\cdot)$ that randomly permutes the messages from all the local randomizers, and (iii) an analyst $A(\cdot)$ that consumes the permuted output of the shuffler; the output of the protocol $P$ is the output of the analyst $A(\cdot)$.  Privacy in the shuffle model is defined as follows:
\begin{definition}[\cite{CheuSUZZ19,erlingsson2019amplification}] 
A protocol $P = (R, S, A)$ is said to be $(\epsilon, \delta)$-$\shuffledDP$ if for any input dataset $D = (x_1, \dots, x_n)$ where $n$ is the number of users, it holds that $S(R(x_1), \dots, R(x_n))$ is $(\epsilon, \delta)$-DP. In the particular case where $\delta = 0$, the protocol $P$ is said to be $\epsilon$-$\shuffledDP$.
\end{definition}

For several analytics tasks, low-error algorithms are known in the central model, whereas they are known to be impossible in the local model. For such tasks, low-error algorithms are commonly sought in the shuffle model, since it is more preferable to trust a shuffler than a central curator.

\subsection{Our Contributions}
\label{sec:our-con}

In the \emph{binary summation} (aka \emph{counting}) problem, each user $i$ receives an input $x_i \in \{0, 1\}$ and the goal is to estimate $\sum_{i \in [n]} x_i$.  For this well-studied task, the discrete Laplace mechanism is known to achieve the optimal (expected absolute) error of $O(1/\epsilon)$ for $\epsilon$-DP summation in the central model \cite{GhoshRS09,GengV16}. Note that this error is independent of the number $n$ of users, and is an absolute constant for the common parameter regime where $\epsilon = O(1)$.  In contrast, the error of any aggregation protocol in the local model is known to be at least on the order of $\sqrt{n}$ \cite{beimel2008distributed, ChanSS12}. There have been many works studied aggregation in the $\shuffledDP$ setting including \cite{BalleBGN19,balle_merged,ghazi2019private,GGKMPV20,GKMP20,GMPS21,GhaziKM21,balcer2019separating}. For pure-DP aggregation, it is known that any single-message protocol (where each user sends a single message to the shuffler) should incur error $\Omega_{\epsilon}(\sqrt{n})$ \cite{balcer2019separating}. For multi-message protocols, where each user can send multiple messages to the shuffler, the best known protocols incur either an error of $O(1/\epsilon^{1.5})$ with  $O(\log{n})$ messages per user~\cite{GGKMPV20} or an error of $O(1/\epsilon)$ with $O(n^{2.5})$  messages per user~\cite{CY21}. No  protocol simultaneously achieved error $O(1/\epsilon)$ and communication $O(\log{n})$. 

We obtain an $\eps$-$\shuffledDP$ algorithm for binary summation, where each user, in expectation, sends $O\left(\frac{\log n}{\eps}\right)$ one-bit messages; this answers the main open question for this basic aggregation task.

\begin{theorem} \label{thm:main-counting}
For every positive real number $\eps \leq O(1)$, there is a (non-interactive) $\eps$-$\shuffledDP$ protocol for binary summation with  RMSE $O(1/\eps)$, where each user sends $O\left(\frac{\log n}{\eps}\right)$ messages in expectation and each message consists of a single bit.
\end{theorem}

In fact, similar to the protocol of~\cite{CY21}, our protocol can get an error that is arbitrarily close to that of the discrete Laplace mechanism, which is known to be optimal in the central model. We defer the formal statement to \Cref{thm:main-counting-nearly-opt-error}.

Before we continue, we note that while the expected number of messages 
in \Cref{thm:main-counting} is small (and with an exponential tail), the \emph{worst} case number of messages is unbounded. This should be contrasted with an  $\Omega_\eps(\sqrt{\log n})$ lower bound in \cite{GGKMPV20} that only applies to the worst case number of bits sent by a user. We discuss this further in \Cref{sec:conclusion}. 

\paragraph{Protocols for Real Summation and Histogram.}

Using known techniques (e.g.,~\cite{CheuSUZZ19,GGKMPV20}), we immediately get the following consequences for real summation and histogram. 

In the \emph{real summation} problem, each $x_i$ is a real value in $[0, 1]$; the goal is again to estimate the sum $\sum_{i \in [n]} x_i$.  The protocol in~\cite{GGKMPV20} achieves an expected RMSE of $\tO(1/\eps^{1.5})$; here, each user sends $O_{\eps}(\log^3 n)$ messages each of length $O(\log \log n)$ bits.
By running their protocol bit-by-bit with an appropriate privacy budget split, we get an algorithm with an improved, and asymptotically optimal, error of $O(1/\eps)$ while with expected communication similar to theirs.

\begin{corollary}
For every positive real number $\epsilon \leq O(1)$, there is a (non-interactive) $\eps$-$\shuffledDP$ protocol for real summation with RMSE $O(1 / \eps)$, where each user sends $O(\frac{\log^3 n}{\eps})$ messages in expectation and each message consists of $O(\log\log n)$ bits.
\end{corollary}

A widely used primitive related, though not identical, to aggregation is histogram computation. 
In the \emph{histogram} problem, each $x_i$ is a number in $[B]$; the goal is to estimate the histogram of the dataset, where the histogram $\bh \in \Z_{\geq 0}^B$ is defined by $h_b = |\{i \in [n] \mid x_i = b\}|$. The error of an estimated histogram $\tbh$ is usually measured in the $\ell_\infty$ sense, i.e., $\|\tbh - \bh\|_{\infty} = \max_{b \in [B]} |h_b - \th_b|$.

For this task, which has been studied  in several papers including \cite{anon-power,balcer2019separating}, the best known pure-$\shuffledDP$ protocol achieved $\ell_{\infty}$-error $O\left(\frac{\log{B} \log{n}}{\eps^{1.5}}\right)$ and communication $O\left(\frac{B\log{n} \log{B}}{\eps}\right)$ bits.  By running the our $(\eps/2)$-$\shuffledDP$ protocol separately for each bucket~\cite[Appendix A]{GGKMPV20}, we immediately arrive at the following:

\begin{corollary}
For every positive real number $\eps \leq O(1)$, there is a (non-interactive) $\eps$-$\shuffledDP$ protocol that computes histograms on domains of size $B$ with an expected $\ell_{\infty}$-error of at most $O\left(\frac{\log{B} \log{n}}{\eps}\right)$, where each user sends $O\left(\frac{B\log{n}}{\eps}\right)$ messages in expectation and each message consists of $O(\log{B})$ bits.
\end{corollary}


\subsection{Technical Overview}

We will now briefly discuss the proof of \Cref{thm:main-counting}. Surprisingly, we show that a simple modification of the algorithm from \cite{GKMP20} satisfies pure-DP! To understand the modification and its necessity, it is first important to understand the algorithm in~\cite{GKMP20}. In their protocol, the messages are either $+1$ or $-1$, and the analyzer's output is simply the sum of all messages. There are three type of messages each user sends:
\begin{itemize}
\item \emph{Input-Dependent Messages}: If the input $x_i$ is 1, the user sends a +1 message. Otherwise, the user does not send anything.
\item \emph{Flooding Messages}: These are messages that do \emph{not} affect the final estimation error. In particular, a random variable $z^{\pm 1}_i$ is drawn from an appropriate distribution and the user sends  $z^{\pm 1}_i$ additional copies of $-1$ and $z^{\pm 1}_i$ additional copies of $+1$. These messages get canceled  when the analyzer computes it output.
\item \emph{Noise Messages}: These are the messages that affect the error in the end. Specifically, $z^{+1}_i, z^{-1}_i$ are drawn i.i.d. from an appropriate distribution, and $z^{- 1}_i$ additional copies of $-1$ and $z^{+ 1}_i$ additional copies of $+1$ are then sent.
\end{itemize}

We note here that the view of the analyzer is simply the number of +1 messages and the number of $-1$ messages, which we will denote by $V_{+1}$ and $V_{-1}$ respectively. 

While~\cite{GKMP20} shows that this protocol is $(\eps, \delta)$-DP, it is easy to show that this is \emph{not} $\eps$-DP for any finite $\eps$. Indeed, consider two neighboring datasets where $X$ consists of all zeros and $X'$ consists of a single one and $n - 1$ zeros. There is a non-zero probability that $V_{+1}(X) = 0$, while $V_{+1}(X')$ is always non-zero (because of the input-dependent message from the user holding the single one). 

To fix this, we randomize this ``input-dependent'' part.  With probability $q$, the user sends nothing.  With the remaining probability $1 - q$, (instead of sending a single +1 for $x_i = 1$ as in~\cite{GKMP20},) the user sends $s + 1$ copies of +1 and $s$ copies of $-1$; similarly, the user sends $s$ copies of +1 and $s$ copies of $-1$ in the case $x_i = 0$. By setting $q$ to be sufficiently small (e.g., $q = O(\nicefrac{1}{\eps n})$), it can be shown that the error remains roughly the same as before. Furthermore, when $s$ is sufficiently large (i.e., $O_\eps(\log n)$), we manage to show that this algorithm satisfies $\eps$-$\shuffledDP$. While the exact reason for this pure-DP guarantee is rather technical, the general idea is similar to~\cite{GGKMPV20}: by making the ``border'' part of the support equal in probabilities in the two cases, we avoid the issues presented above.  Furthermore, by making  $s$ sufficiently large, the input-dependent probability is ``sufficiently inside'' of the support that it usually does not completely dominate the contribution from the outer part.


Finally, note that $V_{+1}, V_{-1}$ involves summation of many i.i.d. random variables $\sum_{i \in [n]} z^{\pm 1}_i$, $\sum_{i \in [n]} z^{+ 1}_i$, and $\sum_{i \in [n]} z^{- 1}_i$. As observed in~\cite{GKMP20}, it is convenient to use \emph{infinitely divisible} distributions so that these sums have distributions that are independent of  $n$, allowing for simpler calculations. We inherit this feature from their analysis.

\section{Preliminaries}

For a discrete distribution $\cD$, let $f_{\cD}$ denote its probability mass function (PMF). The \emph{max-divergence} between distributions $\cD_1, \cD_2$ is defined as $d_{\infty}(\cD_1 \| \cD_2) := \max_{x \in \supp(\cD_1)} \ln(f_{\cD_1}(x) / f_{\cD_2}(x))$.

For two distributions $\cD_1, \cD_2$ over $\Z^d$, we write $\cD_1 * \cD_2$ to denote its \emph{convolution}, i.e., the distribution of $z_1 + z_2$ where $z_1 \sim \cD_1, z_2 \sim \cD_2$ are independent. Moreover, let $(\cD)^{* n}$ denote the $n$-fold convolution of $\cD$, i.e., the distribution of $z_1 + \cdots + z_n$ where $z_1, \dots, z_n \sim \cD$ are independent. We write $\cD \otimes \cD'$ to denote the \emph{product} distribution of $\cD_1, \cD_2$. Furthermore, we may write a value to denote the distribution all of whose probability mass is at that value (e.g., 0 stands for the probability distribution that is always equal to zero).

A distribution $\cD$ is \emph{infinitely divisible} iff, for every positive integer $n$, there exists a distribution $\cD_{/n}$ such that $(\cD_{/n})^{* n}$. 
Two distributions we will use here (both supported on $\Z_{\geq 0}$) are:
\begin{itemize}
\item \emph{Poisson Distribution} $\Poi(\lambda)$: This is the distribution whose PMF is $f_{\Poi(\lambda)}(k) = \lambda^k e^{-\lambda} / k!$. It satisfies $\Poi(\lambda)_{/n} = \Poi(\lambda/n)$.
\item \emph{Negative Binomial Distribution} $\NB(r, p)$: This  is the distribution whose PMF is $f_{\NB(r, p)}(k) = \binom{k+r-1}{k} p^r(1 - p)^k$. It satisfies $\NB(r, p)_{/n} = \NB(r/n, p)$.
\begin{itemize}
\item \emph{Geometric Distribution} $\Geo(p)$: A special case of the $\NB$ distribution is the geometric distribution $\Geo(p) = \NB(1, p)$, i.e., one with $f_{\Geo(p)}(k) = p (1 - p)^k$.
\end{itemize}
\end{itemize}
Finally, we recall that the \emph{discrete Laplace distribution} $\DLap(a)$ is a distribution supported on $\Z$ with PMF $f_{\DLap(a)}(x) \propto \exp\left(-a |x|\right)$. It is well known that $\DLap(a)$ is the distribution of $z_1 - z_2$ where $z_1, z_2 \sim \Geo(1 - \exp(-a))$ are independent. Furthermore, the variance of the discrete Laplace distribution is $\Var(\DLap(a)) = \frac{2e^{-a}}{(1 - e^{-a})^2}$.

We will also use the following well-known lemma\footnote{This can be viewed as a special case of the post-processing property of DP where the post-processing function is adding a random variable drawn from $\cD_3$. Another way to see that this holds is to simply observe that, for any $y \in \supp(\cD_1 * \cD_2)$, we have $f_{\cD_1 * \cD_3}(y) = \sum_{z \in \supp(\cD_3)} f_{\cD_3}(z) \cdot f_{\cD_1}(y - z) \leq \sum_{z \in \supp(\cD_3)} f_{\cD_3}(z) \cdot \left(e^{d_{\infty}(\cD_1 \| \cD_2)} \cdot f_{\cD_2}(y - z)\right) = e^{d_{\infty}(\cD_1 \| \cD_2)} f_{\cD_2 * \cD_3}(y)$.}:
\begin{lemma} \label{lem:conv-reduce-div}
For any distributions $\cD_1, \cD_2, \cD_3$ over $\Z^d$,
$d_{\infty}(\cD_1 * \cD_3 \| \cD_2 * \cD_3) \leq d_{\infty}(\cD_1 \| \cD_2)$.
\end{lemma}

\section{Counting Protocol}

In this section, we will describe a pure-$\shuffledDP$  algorithm for counting, which is our main result. 

\begin{theorem} \label{thm:main-counting-nearly-opt-error}
For any positive real numbers $\eps \leq O(1)$ and $\rho \in (0, 1/2]$, there is a (non-interactive) $\eps$-$\shuffledDP$ protocol for binary summation that has MSE at most $(1 + \rho) \cdot \Var(\DLap(e^{-\eps}))$ where each user sends $O\left(\frac{\log(n/\rho)}{\eps\rho}\right)$ messages in expectation and each message consists of a single bit.
\end{theorem}

By setting $\rho$ arbitrarily close to zero, we can get the mean square error (MSE) to be arbitrarily close to that of the discrete Laplace mechanism, which is known to be (asymptotically) optimal in the central model~\cite{GhoshRS09,GengV16}. We can get this guarantee for other type of errors, e.g., $\ell_1$-error (aka expected absolute error) as well, but for ease of presentation, we only focus on MSE.

Note that \Cref{thm:main-counting-nearly-opt-error} implies \Cref{thm:main-counting} by simply setting $\rho$ to be a positive constant (say, $0.5$).

\subsection{Algorithm}

In this section we present and analyze our main algorithm for counting (aka binary summation).  To begin, we will set our parameters as follows.

\begin{condition} \label{cond:parameters}
Let $\lambda, \eps', \eps, q \in \R_{> 0}$ and $s \in \Z_{> 0}$. 
Suppose that the following conditions hold:
\begin{itemize}
\item $\eps' < \eps$,
\item $s \geq 2\ln\left(\frac{1}{(e^{\eps} - 1)q}\right) / (\eps - \eps')$,
\item $\lambda \geq \frac{e^{\eps - \eps'}}{1 - e^{(\eps' - \eps)/2}} \cdot s$.
\end{itemize}
\end{condition}
We now define the following distributions:
\begin{itemize}
\item $\cDcentral = \Geo(1 - e^{-\eps'})$.
\item $\cDflood = \Poi(\lambda)$.
\item For $x \in \{0, 1\}$, $\cDinp{x}$ supported on $\Z_{\geq 0}^2$ is defined as
\begin{align*}
\cDinp{x}((s + x, s)) &= 1 - q, \\
\cDinp{x}((0, 0)) &= q.
\end{align*}
\end{itemize}

\Cref{alg:delta_randomizer} contains the formal description of the randomizer and \Cref{alg:delta_analyzer} contains the description of the analyzer.  As mentioned earlier, our algorithm is the same as that of \cite{GKMP20}, except in the first step (Line~\ref{line:sample-inp}). In~\cite{GKMP20}, the protocol always sends a single +1 if $x_i = 1$ and nothing otherwise in this step. Instead, we randomize this step by always sending nothing with a certain probability. With the remaining probability, instead of sending a single +1 for $x_i = 1$, we send $s + 1$ copies of +1 and $s$ copies of $-1$ (similarly, we send $s$ copies of +1 and $s$ copies of $-1$ in the case $x_i = 0$).


\begin{algorithm}[h]
\caption{\small Counting Randomizer} \label{alg:delta_randomizer}
\begin{algorithmic}[1]
\Procedure{\randomizer$_n(x_i)$}{}
\State Sample $(y^{+1}_{i}, y^{-1}_{i}) \sim \cDinp{x_i}$ \label{line:sample-inp}
\EndIf
\State Sample $z^{+1}_i, z^{-1}_i \sim \cDcentral_{/n}$
\State Sample $z^{\pm 1}_i \sim \cDflood_{/n}$ 
\State Send $y_i^{+1} + z^{+1}_i + z^{\pm 1}_i$ copies of $+1$, and $y^{-1}_i + z^{-1}_i + z^{\pm 1}_i$ copies of $-1$ \label{line:plusoneminusone} 
\EndProcedure
\end{algorithmic}
\end{algorithm}
\begin{algorithm}[h]
\caption{\small Counting Analyzer} \label{alg:delta_analyzer}
\begin{algorithmic}[1]
\Procedure{\analyzer$_{n, s}$}{}
\State $R \leftarrow$ multiset of messages received
\State \Return $\left(\sum_{y \in R} y\right) - ns$
\EndProcedure
\end{algorithmic}
\end{algorithm}
 
\section{Analysis of the Protocol}

In this section we analyze the privacy, utility, and communication guarantees of our counting protocol. Throughout the remainder of this section, we assume the distributions and parameters are set as in \Cref{cond:parameters}; for brevity, we will not state this assumption in our privacy analysis.

\subsection{Privacy Analysis}
\label{sec:privacy-proof}

\begin{lemma}[Main Privacy Guarantee] \label{lem:dp-main}
\randomizer~satisfies $\eps$-$\shuffledDP$. 
\end{lemma}

To prove the above, we need the following technical lemmas regarding $\cDcentral, \cDflood$.

\begin{lemma} \label{lem:geo-ratio}
For every $i \in \Z$, $f_{\cDcentral}(i - 1) \leq e^{\eps'} f_{\cDcentral}(i)$
\end{lemma}

\begin{proof}
This immediately follows from the PMF definition of $\cDcentral = \Geo(1 - e^{\eps'})$.
\end{proof}

\begin{lemma} \label{lem:poi-ratio}
For every $i \in \Z$, $(e^{\eps} - 1) q \cdot f_{\cDflood}(i + s) + e^{\eps-\eps'} f_{\cDflood}(i - 1) \geq f_{\cDflood}(i)$.
\end{lemma}

\begin{proof}
If $e^{\eps-\eps'} f_{\cDflood}(i - 1) \geq f_{\cDflood}(i)$, then the statement is clearly true. Otherwise, we have $f_{\cDflood}(i) > 0$ (i.e.,  $i \geq 0$) and $e^{\eps' - \eps} > \frac{f_{\cDflood}(i - 1)}{f_{\cDflood}(i)} = \frac{i}{\lambda}$, which implies
\begin{align} \label{eq:range-i}
0 \leq i \leq e^{\eps' - \eps} \lambda.
\end{align}
We can then bound $\frac{f_{\cDflood}(i + s)}{f_{\cDflood}(i)}$ as
\begin{align*}
\frac{f_{\cDflood}(i + s)}{f_{\cDflood}(i)} = \frac{\lambda^s}{(i + 1) \cdots (i + s)}
\geq \frac{\lambda^s}{(i + s)^s}
\overset{\eqref{eq:range-i}}{\geq} \left(\frac{\lambda}{e^{\eps' - \eps}\lambda + s}\right)^s
\geq \left(\frac{\lambda}{e^{(\eps' - \eps)/2}\lambda}\right)^s 
\geq \frac{1}{(e^{\eps} - 1)q},
\end{align*}
where the last two inequalities follow from our assumptions on $\lambda$ and $s$ respectively (\Cref{cond:parameters}). Thus, in this case, we also have $q \cdot \cDflood(i + s) + e^{-\eps'} \cDflood(i - 1) \geq e^{-\eps} \cDflood(i)$ as desired.
\end{proof}

We are now ready to prove the privacy guarantee (\Cref{lem:dp-main}).

\begin{proof}[Proof of \Cref{lem:dp-main}]
For any input dataset $X$. Let $V(X) = (V_{+1}, V_{-1})$ denote the distribution of the view of shuffler, where $V_{+1}$ and $V_{-1}$ denotes the number of +1 messages and the number of $-1$ messages respectively.

Consider two neighboring datasets $X = (x_1, \dots, x_n)$ and $X' = (x'_1, \dots, x'_n)$. Assume w.l.o.g. that they differ in the first coordinate and $x_1 = 1, x'_1 = 0$ and 
$x_2' = x_2$, \ldots, $x_n' = x_n$. To prove that \randomizer~satisfies $\eps$-$\shuffledDP$, we need to prove that $d_\infty(V(X) \| V(X')) \leq \eps$ and $d_\infty(V(X') \| V(X)) \leq \eps$.

Let $\cF$ denote the distribution on $\Z^2$ of $(X, X)$ where $X \sim \cDflood$. Observe that
\begin{align*}
V(X) = \cDinp{1} * \cDinp{x_2} * \cdots * \cDinp{x_n} * \cF * (\cDcentral \otimes 0) * (0 \otimes \cDcentral),
\end{align*}
and
\begin{align*}
V(X') = \cDinp{0} * \cDinp{x_2} * \cdots * \cDinp{x_n} * \cF * (\cDcentral \otimes 0) * (0 \otimes \cDcentral).
\end{align*}

\paragraph{Bounding $d_\infty(V(X) \| V(X'))$.}
From \Cref{lem:conv-reduce-div}, we have
\begin{align*}
d_\infty(V(X) \| V(X')) \leq d_\infty(\cDinp{1} * (\cDcentral \otimes 0) \| \cDinp{0} * (\cDcentral \otimes 0)).
\end{align*}
For any $i, j \in \Z$, we have
\begin{align*}
f_{\cDinp{1} * \cDcentral \otimes 0}(i, j) &= q \cdot f_{\cDcentral}(i) \bone[j = 0] + (1 - q) \cdot f_{\cDcentral}(i - s - 1) \bone[j = s] \\
(\text{\Cref{lem:geo-ratio}}) &\leq q \cdot f_{\cDcentral}(i) \bone[j = 0] + (1 - q) \cdot e^{\eps'} f_{\cDcentral}(i - s) \bone[j = s] \\
(\text{\Cref{cond:parameters}}) &\leq e^{\eps}\left(q \cdot f_{\cDcentral}(i) \bone[j = 0] + (1 - q) \cdot f_{\cDcentral}(i - s) \bone[j = s]\right) \\
&= e^{\eps} \cdot f_{\cDinp{0} * \cDcentral \otimes 0}(i, j).
\end{align*}
Combining the above inequalities, we have $d_\infty(V(X) \| V(X')) \leq \eps$ as desired.

\paragraph{Bounding $d_\infty(V(X') \| V(X))$.}
Again, from \Cref{lem:conv-reduce-div}, we have
\begin{align*}
d_\infty(V(X') \| V(X)) \leq d_\infty(\cDinp{0} * \cF * (0 \times \cDcentral) \| \cDinp{1} * \cF * (0 \times \cDcentral)).
\end{align*}
For any $i, j \in \Z$, we have
\begin{align*}
&f_{\cDinp{0} * \cF * (0 \times \cDcentral)}(i, j) \\
&= f_{\cDinp{0} * \cF}(i, i) \cdot f_{\cDcentral}(j - i) \\
&= \left(q \cdot f_{\cDflood}(i) + (1 - q) \cdot f_{\cDflood}(i - s) \right) \cdot f_{\cDcentral}(j - i) \\
(\text{\Cref{lem:poi-ratio}}) &\leq e^{\eps} \left(q \cdot f_{\cDflood}(i) + (1 - q) \cdot e^{-\eps'} f_{\cDflood}(i - s - 1) \right) \cdot f_{\cDcentral}(j - i) \\
(\text{\Cref{lem:geo-ratio}}) &\leq e^{\eps} \left(q \cdot f_{\cDflood}(i) \cdot f_{\cDcentral}(j - i) + (1 - q) \cdot f_{\cDflood}(i - s - 1) \cdot f_{\cDcentral}(j - i + 1) \right) \\
&= e^{\eps} f_{\cDinp{1} * \cF * (0 \times \cDcentral)}(i, j).
\end{align*}
Combining the above two inequalities, we have $d_\infty(V(X') \| V(X)) \leq \eps$, concluding our proof.
\end{proof}

\subsection{Utility Analysis}

We next analyze the MSE of the output estimate.

\begin{lemma} \label{lem:util}
The MSE of the estimator is at most $\Var(\DLap(e^{-\eps'})) + qn + q^2n(n - 1)$.
\end{lemma}

\begin{proof}
Notice that the output estimate is equal to $\sum_{i \in [n]} (y^{+1}_i - y^{-1}_i + z^{+1}_i - z^{-1}_i) = \sum_{i \in [n]} (y^{+1}_i - y^{-1}_i) + Z$ where $Z \sim \DLap(e^{-\eps'})$. As a result, the MSE of the output estimate is equal to
\begin{align*}
\E\left[\left(\sum_{i \in [n]} (y^{+1}_i - y^{-1}_i - x_i) + Z\right)^2\right] = \E\left[\left(\sum_{i \in [n]} (y^{+1}_i - y^{-1}_i - x_i)\right)^2\right] + \Var(\DLap(e^{-\eps'})). 
\end{align*}
Next, notice that, if $x_i = 0$, then $y^{+1}_i - y^{-1}_i - x_i = 0$ always. Otherwise, if $x_i = 1$, then $y^{+1}_i - y^{-1}_i - x_i = 0$ with probability $1 - q$ and $y^{+1}_i - y^{-1}_i - x_i = 1$ with probability $q$. As a result, we have
\begin{align*}
\E\left[\left(\sum_{i \in [n]} (y^{+1}_i - y^{-1}_i - x_i)\right)^2\right] &\leq qn + q^2n(n - 1). \qedhere 
\end{align*}
\end{proof}

\subsection{Communication Analysis}

The expected number of bits send by the users can be easily computed as follows.

\begin{lemma} \label{lem:comm}
The expected number of messages sent by each user is at most $2s + 1 + \frac{\lambda}{n} + O\left(\frac{1}{\eps' n}\right)$.
\end{lemma}

\begin{proof}
The expected number of bits sent per user is
\begin{align*}
\E[y^{+1}_i + y^{-1}_i] + \E[z^{+1}_i + z^{-1}_i] + 2\E[z^{\pm 1}_i] 
&\leq (2s + 1) + \frac{2\E[\cDcentral]}{n} + \frac{\E[\cDflood]}{n} \\ 
&= 2s + 1 + O\left(\frac{1}{\eps' n}\right) + \frac{\lambda}{n}. \qedhere
\end{align*}
\end{proof}

\subsection{Putting Things Together: Proof of \Cref{thm:main-counting-nearly-opt-error}}

Finally, we are ready to prove \Cref{thm:main-counting-nearly-opt-error} by plugging in appropriate parameters and invoke the previous lemmas.

\begin{proof}[Proof of \Cref{thm:main-counting-nearly-opt-error}]
We start by picking $\eps' = \eps - 0.01 \rho \cdot \min\{ \eps, 1\}$. For this choice of $\eps'$, we have
\begin{align*}
\frac{\Var(\DLap(\eps'))}{\Var(\DLap(\eps))}
= \frac{\frac{2e^{-\eps'}}{(1 - e^{-\eps'})^2}}{\frac{2e^{-\eps}}{(1 - e^{-\eps})^2}}
\leq 1 + \frac{(e^{\eps - \eps'} - 1)(1 + e^{-\eps'})}{1 - e^{-\eps'}} 
\leq 1 + \frac{3(\eps - \eps') \cdot 2}{\eps'}
\leq 1 + 0.5\rho.
\end{align*}

Then, picking
\begin{align*}
q &= \frac{0.1 \rho \cdot \Var(\DLap(\eps))}{n} = O\left(\frac{\rho}{\eps^2 n}\right), \\ 
s &\geq 2\ln\left(\frac{1}{(e^{\eps} - 1)q}\right) / (\eps - \eps') = O\left(\frac{\log(n/\rho)}{\eps\rho}\right), \\
\lambda &\geq \frac{e^{\eps - \eps'}}{1 - e^{(\eps' - \eps)/2}} \cdot s = O\left(\frac{\log(n/\rho)}{\eps^2 \rho}\right),
\end{align*}
and applying \Cref{lem:dp-main}, \Cref{lem:util}, and \Cref{lem:comm} immediately imply \Cref{thm:main-counting-nearly-opt-error}. (Note that we may assume that $\eps \geq 1/n$; otherwise we can just output zero. Under this assumption, we have $\lambda / n \leq O\left(\frac{\log(n/\rho)}{\eps\rho}\right)$ as desired for the communication complexity claim.)
\end{proof}

\section{Conclusions and Open Questions}
\label{sec:conclusion}

In this work, we have provided pure-$\shuffledDP$ algorithms that achieve nearly optimal errors for bit summation, real summation, and histogram while significantly improving on the communication complexity compared to the state-of-the-art. Despite this, there are still a number of interesting open questions, some of which we highlight below.

\begin{itemize}
\item \textbf{Protocol with a bounded number of messages.} As mentioned briefly in \Cref{sec:our-con}, our protocol can result in an arbitrarily large number of messages per user, although the expected number is quite small. (In fact, the distribution of the number of messages enjoys a strong exponential tail bound.) Is it possible to design a pure-$\shuffledDP$ protocol where the maximum number of messages is $O\left(\frac{\log n}{\eps}\right)$ for binary summation?

For this question, we note that a rather natural approach is to modify our protocol to make its number of messages bounded. Namely, we replace $\cDcentral_{/n}$ and $\cDflood_{/n}$ by a truncated version of their respective distributions. It turns out that the latter is relatively simple (e.g., even replacing it with a Bernoulli distribution also works) because we only require a mild condition in \Cref{lem:poi-ratio} to hold. On the other hand, for the former, we are using \Cref{lem:geo-ratio} which only holds for unbounded distributions. We would like to stress that we do not know whether replacing $\cDcentral_{/n}$ with a truncated version of the negative binomial distribution with a ``symmetrized'' the input dependent part\footnote{This means that w.p. $q$ we output $s$ copies of both $+1$ and $-1$ messages, for both $x_i = 0$ and $x_i = 1$ cases. Without this change, the supports of the two cases are not the same and thus it obviously violates pure-DP.} violates pure-DP; however, we do not know how to prove that it satisfies pure-DP either, as the probability mass function of their convolutions become somewhat unwieldy.
\item \textbf{Lower bounds on the expected number of messages.} Recall that the communication lower bound from \cite{GGKMPV20} only applies to the maximum number of messages sent.  Is it possible to prove a communication lower bound on the \emph{expected} number of messages (even if the maximum number of messages is unbounded)?  We note that the techniques from \cite{GGKMPV20} does not apply.

\item \textbf{Histogram protocol for large $B$.} Our protocol has communication complexity that grows linearly with $B$, which becomes impractical when $B$ is large. Can we get protocol for histogram whose communication is $O_\eps\left((\log n)^{O(1)}\right)$ for $B = O(n)$ (while achieving nearly optimal errors)? For approximate-$\shuffledDP$, a histogram protocol with expected communication of $1 + O_\eps\left(\frac{B (\log(n/\delta)^{O(1)})}{n}\right)$ is known~\cite{GKMP20}. It would be interesting to understand if extending such a protocol to the 
 pure-$\shuffledDP$ setting is possible.
\end{itemize}

More generally, despite the (by-now) vast literature on the shuffle model, most work have focused attention on approximate-$\shuffledDP$. It would be interesting to expand the existing study to pure-$\shuffledDP$ as well.

\bibliographystyle{alpha}
\bibliography{ref}

\newcommand{\etalchar}[1]{$^{#1}$}
\begin{thebibliography}{GKM{\etalchar{+}}21b}

\bibitem[BBGN19]{BalleBGN19}
Borja Balle, James Bell, Adri{\`{a}} Gasc{\'{o}}n, and Kobbi Nissim.
\newblock The privacy blanket of the shuffle model.
\newblock In {\em CRYPTO}, pages 638--667, 2019.

\bibitem[BBGN20]{balle_merged}
Borja Balle, James Bell, Adri{\`{a}} Gasc{\'{o}}n, and Kobbi Nissim.
\newblock Private summation in the multi-message shuffle model.
\newblock In {\em {CCS}}, pages 657--676, 2020.

\bibitem[BC20]{balcer2019separating}
Victor Balcer and Albert Cheu.
\newblock Separating local {\&} shuffled differential privacy via histograms.
\newblock In {\em {ITC}}, pages 1:1--1:14, 2020.

\bibitem[BEM{\etalchar{+}}17]{bittau17}
Andrea Bittau, {\'{U}}lfar Erlingsson, Petros Maniatis, Ilya Mironov, Ananth
  Raghunathan, David Lie, Mitch Rudominer, Ushasree Kode, Julien Tinn{\'{e}}s,
  and Bernhard Seefeld.
\newblock Prochlo: Strong privacy for analytics in the crowd.
\newblock In {\em SOSP}, pages 441--459, 2017.

\bibitem[BNO08]{beimel2008distributed}
Amos Beimel, Kobbi Nissim, and Eran Omri.
\newblock Distributed private data analysis: Simultaneously solving how and
  what.
\newblock In {\em CRYPTO}, pages 451--468, 2008.

\bibitem[CSS12]{ChanSS12}
T.{-}H.~Hubert Chan, Elaine Shi, and Dawn Song.
\newblock Optimal lower bound for differentially private multi-party
  aggregation.
\newblock In {\em ESA}, pages 277--288, 2012.

\bibitem[CSU{\etalchar{+}}19]{CheuSUZZ19}
Albert Cheu, Adam~D. Smith, Jonathan Ullman, David Zeber, and Maxim Zhilyaev.
\newblock Distributed differential privacy via shuffling.
\newblock In {\em EUROCRYPT}, pages 375--403, 2019.

\bibitem[CY22]{CY21}
Albert Cheu and Chao Yan.
\newblock Pure differential privacy from secure intermediaries.
\newblock In {\em TPDP}, 2022.

\bibitem[DKM{\etalchar{+}}06]{dwork2006our}
Cynthia Dwork, Krishnaram Kenthapadi, Frank McSherry, Ilya Mironov, and Moni
  Naor.
\newblock Our data, ourselves: Privacy via distributed noise generation.
\newblock In {\em EUROCRYPT}, pages 486--503, 2006.

\bibitem[DMNS06]{dwork2006calibrating}
Cynthia Dwork, Frank McSherry, Kobbi Nissim, and Adam Smith.
\newblock Calibrating noise to sensitivity in private data analysis.
\newblock In {\em TCC}, pages 265--284, 2006.

\bibitem[EFM{\etalchar{+}}19]{erlingsson2019amplification}
{\'U}lfar Erlingsson, Vitaly Feldman, Ilya Mironov, Ananth Raghunathan, Kunal
  Talwar, and Abhradeep Thakurta.
\newblock Amplification by shuffling: From local to central differential
  privacy via anonymity.
\newblock In {\em SODA}, pages 2468--2479, 2019.

\bibitem[EGS03]{evfimievski2003limiting}
Alexandre Evfimievski, Johannes Gehrke, and Ramakrishnan Srikant.
\newblock Limiting privacy breaches in privacy preserving data mining.
\newblock In {\em PODS}, pages 211--222, 2003.

\bibitem[GGK{\etalchar{+}}20]{GGKMPV20}
Badih Ghazi, Noah Golowich, Ravi Kumar, Pasin Manurangsi, Rasmus Pagh, and
  Ameya Velingker.
\newblock Pure differentially private summation from anonymous messages.
\newblock In {\em {ITC}}, pages 15:1--15:23, 2020.

\bibitem[GGK{\etalchar{+}}21]{anon-power}
Badih Ghazi, Noah Golowich, Ravi Kumar, Rasmus Pagh, and Ameya Velingker.
\newblock On the power of multiple anonymous messages: Frequency estimation and
  selection in the shuffle model of differential privacy.
\newblock In {\em {EUROCRYPT}}, pages 463--488, 2021.

\bibitem[GKM21a]{GhaziKM21}
Badih Ghazi, Ravi Kumar, and Pasin Manurangsi.
\newblock User-level differentially private learning via correlated sampling.
\newblock In {\em {NeurIPS}}, pages 20172--20184, 2021.

\bibitem[GKM{\etalchar{+}}21b]{GMPS21}
Badih Ghazi, Ravi Kumar, Pasin Manurangsi, Rasmus Pagh, and Amer Sinha.
\newblock Differentially private aggregation in the shuffle model: Almost
  central accuracy in almost a single message.
\newblock In {\em {ICML}}, pages 3692--3701, 2021.

\bibitem[GKMP20]{GKMP20}
Badih Ghazi, Ravi Kumar, Pasin Manurangsi, and Rasmus Pagh.
\newblock Private counting from anonymous messages: Near-optimal accuracy with
  vanishing communication overhead.
\newblock In {\em {ICML}}, pages 3505--3514, 2020.

\bibitem[GMPV20]{ghazi2019private}
Badih Ghazi, Pasin Manurangsi, Rasmus Pagh, and Ameya Velingker.
\newblock Private aggregation from fewer anonymous messages.
\newblock In {\em {EUROCRYPT}}, pages 798--827, 2020.

\bibitem[GRS09]{GhoshRS09}
Arpita Ghosh, Tim Roughgarden, and Mukund Sundararajan.
\newblock Universally utility-maximizing privacy mechanisms.
\newblock In {\em {STOC}}, pages 351--360, 2009.

\bibitem[GV16]{GengV16}
Quan Geng and Pramod Viswanath.
\newblock The optimal noise-adding mechanism in differential privacy.
\newblock {\em {IEEE} Trans. Inf. Theory}, 62(2):925--951, 2016.

\bibitem[KLN{\etalchar{+}}08]{kasiviswanathan2008what}
Shiva~Prasad Kasiviswanathan, Homin~K. Lee, Kobbi Nissim, Sofya Rashkodnikova,
  and Adam Smith.
\newblock What can we learn privately?
\newblock In {\em FOCS}, pages 531--540, 2008.

\end{thebibliography}



\end{document}